\documentclass[conference,onecolumn]{IEEEtran}

\IEEEoverridecommandlockouts                              

\usepackage{amsmath}
\usepackage{amssymb}
\usepackage{amsthm}
\usepackage{hyperref}
\usepackage{graphicx}
\usepackage{ifthen}
\usepackage{enumerate}
\usepackage{caption}
\usepackage{subcaption}
\usepackage{color}
\usepackage{comment}
\usepackage{epstopdf}

\def\cX{\mathcal{X}}
\def\cY{\mathcal{Y}}

\def\cP{\mathcal{P}}

\def\cV{\mathcal{V}}

\def\cS{\mathcal{S}}
\def\cF{\mathcal{F}}

\def\cE{\mathcal{E}}

\newtheorem{theorem}{Theorem}

\newtheorem{lemma}{Lemma}

\newcommand{\eqdef}{\stackrel{\scriptscriptstyle \triangle}{=}}

\newcommand{\EE}{\mathbb{E}}

\title{Variable-length codes for channels with memory and feedback: error-exponent upper bounds}

\author{Achilleas Anastasopoulos and Jui Wu
	\thanks{The authors are with the Department
	of Electrical Engineering and Computer Science, University of Michigan, Ann
	Arbor, MI, 48105 USA e-mail: { \{juiwu, anastas\}@umich.edu}}
	}

\begin{document}
\maketitle
\begin{abstract}
The reliability function of memoryless channels with noiseless feedback and variable-length coding
has been found to be a linear function of the average rate in the classic work of Burnashev.
In this work we consider unifilar channels with noiseless feedback and study upper bounds for the channel reliability function with variable length codes. In unifilar channels the channel state is known to the transmitter but is unknown to the receiver.
We generalize Burnashev's analysis and derive a similar expression which is linear in average rate and depends on the channel capacity, as well as an additional parameter which relates to a sequential binary hypothesis testing problem over this channel. This parameter is evaluated by setting up an appropriate Markov decision process (MDP). Furthermore, an upper bound for this parameter is derived using a simplified MDP.
Numerical evaluation of the parameter for several binary input/state/output unifilar channels hints at the optimal transmission strategies. Such strategies are studied in a companion paper to provide lower (achievable) bounds on the channel reliability function.
\end{abstract}

\section{Introduction}
\label{sec:intro}

Error exponent analysis has been an active area of research for quite a few decades.
The vast literature in this area can be categorized based on (i) whether the channel is memoryless  or with memory; (ii) whether there is or is not channel output feedback to the transmitter; (iii) whether the employed coding is fixed-length or variable-length; and (iv) whether upper (converse) or lower (achievable) bounds are analyzed.

In the case of memoryless channels with noiseless feedabck Schalkwijk and Kailath~\cite{ScKa66} proposed a transmission scheme for the additive white Gaussian noise (AWGN) channel with infinite error exponent.
On the other hand, Dobrushin~\cite{Do62} and later Haroutunian~\cite{Ha77}, by deriving an error upper bound for discrete memoryless channels (DMCs) showed that at least for symmetric channels there is no gain to be expected through feedback when fixed-length codes are employed.
This was a strong negative result since it suggested that for DMC channels, noiseless feedback can neither improve capacity (as was well known) nor can it improve the error exponent when fixed-length codes are used.
A remarkable result was derived by Burnashev in~\cite{Bu76}, where error exponent matching upper and lower bounds were derived for DMCs with feedback and  variable-length codes. The error exponent has a simple form $E(\overline{R})=C_1(1-\overline{R}/C)$, where $\overline{R}$ is the average rate, $C$ is the channel capacity and $C_1$ is the maximum divergence that can be obtained in the channel for a binary hypothesis testing problem.
Berlin et al~\cite{BeNaRiTe09} have provided a simpler derivation of the Burnashev bound that emphasizes the link between the constant $C_1$ and the binary hypothesis testing problem.
Several variable-length transmission schemes have been proposed in the literature for DMCs and their error exponents have been analyzed~\cite{Ho63, YaIt79, ShFe11}.

In the case of channels with memory and feedback, the capacity was studied in~\cite{TaMi09, PeCuVaWe08, BaAn10}, and
a number of capacity-achievable schemes have been recently studied in the literature~\cite{CoYuTa09, BaAn10, An12a}.
The only work that studies error exponents for variable-length codes for channels with memory and feedback is~\cite{CoYuTa09} where the authors consider finite state channels with channel state known causally to both the transmitter and the receiver.

In this work, we consider channels with memory and feedback, and derive a straight-line upper bound on the error-exponent for variable-length codes. We specifically look at unifilar channels since for this family, the capacity has been characterized in an elegant way through the use of Markov decision processes (MDPs)~\cite{PeCuVaWe08}.
Our technique is motivated by that of~\cite{Bu76}, i.e., studying the rate of decay of the posterior message entropy using martingale theory in two distinct regimes: large and small message entropies.
A major difference between this work as compared to~\cite{Bu76} is that we analyze the multi-step drift behavior of the communication system instead of the one-step drift that is analyzed for DMCs. This is necessitated by the fact that one-step analysis cannot capture the memory inherent in the channel and thus results in extremely loose bounds. It is not surprising that the parameter $C_1$ in our case also relates to the maximum discrimination that can be achieved in this channel in a  binary hypothesis testing problem. In order to evaluate this quantity, we formulate two MDPs with decreasing degree of complexity, the solutions of which are upper bounds on the quantity $C_1$, with the former being tighter than the latter. The tightness of the bounds is argued based on the fact that asymptotically this is the expected performance of the best system, and by achievability results presented in the companion paper~\cite{AnWu17}. An additional contribution of this work is a complete reworking of some of the more opaque proofs of~\cite{Bu76} resulting in significant simplification of the exposition. We finally provide some numerical results for a number of interesting unifilar channels including the trapdoor, chemical, and other two-input/output/state unifilar channels.
The main difference between our work and that in~\cite{CoYuTa09} is that for unfilar channels, the channel state is not observed at the receiver. This complicates the analysis considerably as is evidenced by the different approaches in evaluating the constant $C_1$ in these two works. Furthermore, our results indicate that optimal policies for achieving the maximum divergence are very different when the receiver knows or does not know the channel state.

The remaining part of this paper is organized as follows. In section~\ref{sec:model}, we describe the channel model for the unifilar channel and the class of encoding and decoding strategies. In section~\ref{sec:bound}, we analyze the  drifts of the posterior message entropy in the large- and small-entropy regime. In section~\ref{sec:C_1}, we formulate two MDPs in order to study the problem of one-bit transmission over this channel.
Section~\ref{sec:example} presents numerical results for several unifilar channels. Final conclusions are given in section~\ref{sec:conclusions}.

\section{Channel Model and preliminaries}
\label{sec:model}

Consider a family of finite-state point-to-point channels with inputs $X_t\in\cX$, output $Y_t\in\cY$ and state $S_t\in\cS$ at time $t$, with all alphabets being finite and initial state $S_1=s_1$ known to both the transmitter and the receiver.
The channel conditional probability is
\begin{equation}
P(Y_t,S_{t+1}|X^t, Y^{t-1}, S^t) = Q(Y_t|X_t, S_t) \delta_{g(S_t,X_t,Y_t)}(S_{t+1}),
\end{equation}
for a given stochastic kernel $Q\in \cX\times\cS\rightarrow \cP(\cY)$ and deterministic function $g\in \cS\times\cX \times\cY\rightarrow \cS$, where $\cP(\cY)$ denotes the space of all probability measure on $\cY$, and $\delta_{a}(\cdot)$ is the Kronecker delta function centered at $a$.
This family of channels is referred to as unifilar channels~\cite{PeCuVaWe08}. The authors in~\cite{PeCuVaWe08} have derived the capacity $C$ under certain conditions in the form of
\begin{equation}
C =  \lim_{N\rightarrow \infty}\sup_{(p(x_t|s_t,y^{t-1},s_1))_{t\geq 1}}\frac{1}{N}\sum_{t=1}^{N} I(X_t,S_t;Y_t|Y^{t-1},S_1). \label{eq:capacity}
\end{equation}
In this paper, we restrict our attention to such channels with strictly positive $Q(y|x,s)$ for any $(y,x,s)\in \cY \times \cX \times \cS$ and ergodic behavior so that the above limit indeed exists.
Let $W\in\{1,2,3,\cdots,M=2^K\}\eqdef [M]$ be the message to be transmitted.
In this system, the transmitter receives perfect feedback of the output with unit delay and decides the input $X_t$ based on $(Y^{t-1},S_1)$ at time $t$.
The transmitter can adopt randomized encoding strategies, where $X_t \sim e_t(\cdot|W,Y^{t-1},S_1)$ with a collection of distributions $(e_t: [M] \times \cY^{t-1} \times \cS \rightarrow \cP(\cX))_{t\geq 1}$.
Without loss of generality we can represent the randomized encoder through deterministic mappings $(e_t: [M] \times \cY^{t-1} \times \cS \times \cV \rightarrow \cX)_{t\geq 1}$ with
$X_t = e_t(W,Y^{t-1},S_1,V_t)$ involving the random variables $(V_t)_{t\geq 1}$ which are generated as
$P(V_t|V^{t-1},Y^{t-1},X^{t-1},S^{t-1},W)=P_V(V_t)$. Furthermore, since we are interested in error exponent upper bounds, we can assume that the random variables $V_t$ are causally observed common information among the transmitter and receiver.
The decoding policy consists of a sequence of decoding functions $(d_t: (\cY\times \cV)^t \times \cS \rightarrow [M])_{t\geq 1}$, with estimated message at every time $t$, $\hat{W}_t = d_t(Y^t,V^t,S_1)$, and a stopping time $T$ w.r.t. the filtration $(\cF_t\triangleq \sigma(Y^t,V^t,S_1))_{t\geq 0}$. The final message estimate is defined as $\hat{W}=\hat{W}_T$.
The average rate $\overline{R}$ and error probability $P_e$ of this scheme are defined as $\overline{R}=\frac{K}{\EE[T]}$ and $P_e = P(\hat{W}_T \neq W \cup T= \infty)$.
The channel reliability function (highest achievable error exponent) is defined as $E^*(\overline{R})=\sup -\frac{\log P_e}{\EE[T]}$.
Since transmission schemes with $P(T=\infty)>0$ result in the trivial error exponent $-\frac{\log P_e}{\EE[T]}=0$, we restrict attention to those schemes that have a.s. finite decision times.

\section{Error-exponent upper bound}\label{sec:bound}
Our methodology is inspired by the analysis in~\cite{Bu76} for DMCs.
The analysis involves lower-bounding the rate of decrease of the posterior message entropy which, through a generalization of Fano's Lemma, provides lower bounds on the error probability. Entropy can decrease no faster than the channel capacity. However this bound becomes trivial at low values of entropy which necessitates switching to lower bounding the corresponding log drift. The log drift analysis is quite involved in~\cite{Bu76} even for the DMC.
The fundamental difference in our work compared to DMC, is the presence of memory in unifilar channels.
A single-step drift analysis wouldn't be able to capture this memory resulting in loose bounds. For this reason we analyze multi-step drifts; in fact we consider the asymptotic behavior as the step size becomes larger and larger.
The outline of the analysis is as follows. Lemma~\ref{lemma:driftentropy} and Lemma~\ref{lemma:driftlogentropy} describe the overall decreasing rate of the entropy induced by the posterior belief on the message in terms of drifts in the linear and logarithmic regime, respectively. The former relates the drift to capacity, $C$, while the latter relates it to a quantity $C_1$ which can be interpreted as the largest discrimination that can be achieved in this channel for a binary hypothesis testing problem, as elegantly explained in~\cite{BeNaRiTe09}.
The result presented in Lemma~\ref{lemma:newsubmartingale} shows that based on a general random process that satisfies the two above mentioned drift conditions one can create an appropriate submartingale.
These three results are then combined together in Theorem~\ref{th:main} to provide a lower bound on the stopping time of an arbitrary system employing variable-length coding, and equivalently an upper bound on the error exponent.

Let us define the following random processes
\begin{align}
\Pi_t(i) &=  P(W=i|\cF_{t}) , \qquad i\in [M], t\geq 0\\
H_t &=  -\sum_{i=1}^{M}\Pi_t(i)\log \Pi_t(i).
\end{align}

Denoting by $h(\cdot)$ the binary entropy function, we have the following result.
\begin{lemma}~\label{lemma:genFano}
Generalized Fano's inequality: If $P(T<\infty)=1$ then
\begin{align}
\EE[H_T] \leq h(P_e) + P_e \log (M-1).
\end{align}
\end{lemma}
\begin{proof}
Please refer to Appendix~\ref{app:genFano}.
\end{proof}
This is essentially~\cite[Lemma 1]{Bu76}
and the proof is presented here for completeness.
In view of the above, in order to estimate the rate of decrease of $P_e$, we study the corresponding rate for $H_t$. The next lemma gives a first estimate of the drift of $(H_t)_{t\geq 0}$.

\begin{lemma} \label{lemma:driftentropy}
For any $t\geq 0$ and $\epsilon>0$, there exists an $N=N(\epsilon)$ such that
\begin{equation}
\EE[H_{t+N} - H_{t}|\mathcal{F}_{t}] \geq -N(C+\epsilon) \qquad a.s.
\end{equation}
\end{lemma}

\begin{IEEEproof}
Please see appendix~\ref{app:lemma1}.
\end{IEEEproof}

Since for small values of $H_t$ the above result does not give any information, we now analyze the drifts of the process $(\log H_t)_{t\geq 0}$.

\begin{lemma} \label{lemma:driftlogentropy}
For any given  $\epsilon>0$, there exists an $N=N(\epsilon)$ such that if $H_t<\epsilon$
\begin{equation}
\EE[\log(H_{t+N})-\log(H_t)|\mathcal{F}_t] \geq   -N(C_1+\epsilon) \qquad a.s.
\end{equation}
where the constant $C_1$ is given by
\begin{align}\label{eq:defC1}
&C_1 =\max_{s_1,y^t,v^t,k}\limsup_{N'\rightarrow \infty} \max_{(e_i)_{i=t+1}^{t+N'}}\frac{1}{N'}
\nonumber \\
&\sum_{Y^{t+N'}_{t+1},V^{t+N'}_{t+1}} P(Y^{t+N'}_{t+1},V^{t+N'}_{t+1}|W=k,y^t,v^t,s_1) \log
\frac{P(Y^{t+N'}_{t+1},V^{t+N'}_{t+1}|W=k,y^t,v^t,s_1)}{ P(Y^{t+N'}_{t+1},V^{t+N'}_{t+1}|W\neq k,y^t,v^t,s_1)}.
\end{align}
\end{lemma}

\begin{IEEEproof}
Please see appendix~\ref{app:lemma2}.
\end{IEEEproof}

We comment at this point that the proof of this result is significantly simpler than the corresponding one  in~\cite[Lemma~3]{Bu76}. The reason is that we develop the proof directly in the asymptotic regime and thus there is no need for complex convexity arguments as the ones derived in~\cite[Lemma~7, and eq. (A8)-(A12)]{Bu76}.
At this point one can bound the quantity in~\eqref{eq:defC1} by $\max_{x,s,x's'} D(Q(y|x,s)||Q(y|x',s')) $ using convexity. Such a bound, however, can be very loose since it does not account for channel memory. In Section~\ref{sec:C_1}, we will discuss how to evaluate $C_1$.
Before we continue, we also note that $|\log H_{t+1} - \log H_t|$ is bounded above by a positive number $C_2$ almost surely due to the fact that kernel $Q(\cdot|\cdot,\cdot)$ is strictly positive. The proof is similar to that in~\cite[Lemma~4]{Bu76}.

In the following lemma, we propose a submartingale that connects drift analysis and the stopping time in the proof of our main result.

\begin{lemma} \label{lemma:newsubmartingale}
Suppose a random process $(H_t)_{t\geq 0}$ has the following properties
\begin{subequations}
\begin{align}
\EE[H_{t+1}-H_{t}|\mathcal{F}_t] &\geq -K_1 \label{ineq:entropy}\\
\EE[\log H_{t+1}-\log H_{t}|\mathcal{F}_t] &\geq -K_2 \label{ineq:logentropysmall}  \qquad \text{if } H_t<H^* \\
|\log H_{t+1}-\log H_{t}| &< K_3  \label{ineq:logentropyall} \qquad \text{if } H_t<H^*
\end{align}
\end{subequations}
almost surely for some positive numbers $K_1,K_2,K_3,H^*$ where $K_2>K_1$.

Define a process $(Z_t)_{t\geq 0}$ by
\begin{align} \label{def:newsubmartingale}
Z_t &= (\frac{H_t-H^*}{K_1}+t)1_{\{H_t>H^*\}} \nonumber \\
 &\ \ + (\frac{\log \frac{H_t}{H^*}}{K_2}+t+f(\log \frac{H_t}{H^*}))1_{\{H_t\leq H^*\}} \quad \forall t\geq 0,
\end{align}
where $f: \mathbb{R} \rightarrow \mathbb{R} $ is defined by
\begin{equation}
f(y) = \frac{1-e^{\lambda y}}{K_2\lambda}
\end{equation}
with a positive constant $\lambda$. Then, for sufficiently small $\lambda$, $(Z_t)_{t\geq 0}$ is a submartingale w.r.t. $\mathcal{F}_t$.
\end{lemma}

\begin{IEEEproof}
Please see appendix~\ref{app:lemma3}.
\end{IEEEproof}

Two comments are in place regarding the proof of this result. First, the main difficulty in proving such results is to take care of what happens in the ``transition'' range (around $H^*$) where $H_t$ and $H_{t+1}$ are not both above or below the threshold. The choice of the function $f(\cdot)$ is what makes the proof work. The proof offered here is quite concise compared to the one employed in~\cite{Bu76} (which consists of Lemma~5 and an approximation argument given in Theorem~1). The reason for that is the specific definition of the $Z_t$ process and in particular the choice of the $f(\cdot)$ function which simplifies considerably the proof. The second, and related, comment is that this Lemma is not a straightforward extension of~\cite[Lemma in p.~50]{BuZi75} since there, the purpose was to bound from below a positive rate of increase of a process. In our case, the proof hinges on the additional constraint~\eqref{ineq:subdiff} we impose on the choice of the $f(\cdot)$ function.

We are now ready to present our main result.
\begin{theorem} \label{th:main}
Any transmission scheme with $M=2^K$ messages and error probability $P_e$  satisfies
\begin{equation}
-\frac{\log P_e}{\EE[T]} \leq C_1(1-\frac{\overline{R}}{C})+ U(\epsilon,K,P_e,\overline{R},C,C_1,C_2,\lambda),
\end{equation}
for any $\epsilon>0$. Furthermore, $\lim_{P_e\rightarrow0}\lim_{K\rightarrow\infty}U(\epsilon,K,\overline{R},C,C_1,C_2,\lambda)=o_{\epsilon}(1)$.
\end{theorem}

\begin{IEEEproof}
Please see appendix~\ref{app:proposition}.
\end{IEEEproof}

\section{Evaluation of $C_1$}
\label{sec:C_1}
In this section we evaluate the constant $C_1$. As noted in \cite{Bu76}\cite{BeNaRiTe09}, the quantity $C_1$ relates to a binary hypothesis testing problem. When the posterior entropy is small, the receiver has very high confidence in a certain message. In this situation, the transmitter is essentially trying to inform the receiver whether or not this candidate message is the true one.
Since the unifilar channel has memory, it is not surprising that the constant $C_1$ may be connected to a Markov decision process related to the aforementioned binary hypothesis testing problem, as was the case in~\cite{CoYuTa09}.
%
Recall that $C_1$ is defined as
\begin{align}
C_1 &= \max_{s_1,y^t,v^t,k}  \limsup_{N'\rightarrow \infty} \max_{\{e_i\}_{i=t+1}^{t+N'}} \frac{1}{N'} \sum_{Y^{t+N'}_{t+1},V^{t+N'}_{t+1}} \nonumber \\
 &P(Y^{t+N'}_{t+1},V^{t+N'}_{t+1}|W=k,y^t,v^t,s_1) \log
\frac{P(Y^{t+N'}_{t+1},V^{t+N'}_{t+1}|W=k,y^t,v^t,s_1)}{ P(Y^{t+N'}_{t+1},V^{t+N'}_{t+1}|W\neq k,y^t,v^t,s_1)} \\
 &= \max_{s_1,y^t,v^t,k}  \limsup_{N'\rightarrow \infty} \max_{\{e_i\}_{i=t+1}^{t+N'}} \frac{1}{N'} \nonumber \\
 & D( P(Y^{t+N'}_{t+1},V^{t+N'}_{t+1}|W=k,y^t,v^t,s_1) ||  P(Y^{t+N'}_{t+1},V^{t+N'}_{t+1}|W\neq k,y^t,v^t,s_1) )
 \label{eq:divergence}
\end{align}

We now look into the quantities $P(Y^{t+N}_{t+1},V^{t+N}_{t+1}|W=k,y^t,v^t,s_1)$ and $P(Y^{t+N}_{t+1},V^{t+N}_{t+1}|W\neq k,y^t,v^t,s_1)$.
Let us define $X^k_{t} \eqdef e_{t}(k,Y^{t-1},S_1,V_t)$ and $S^k_{t} \eqdef g_{t}(k,Y^{t-1},V^{t-1},S_1)$ which are the input and the state at time $t$, respectively conditioned on $W=k$. Then,
\begin{align}
P(Y_{t+1}^{t+N},V^{t+N}_{t+1}|y^t,v^t,s_1,W=k)
&= \prod_{i=t+1}^{t+N} P(Y_i|Y^{i-1}_{t+1},V^{i}_{t+1},y^t,v^t,s_1,W=k)P(V_i|Y^{i-1}_{t+1},V^{i-1}_{t+1},y^t,v^t,s_1,W=k)\nonumber \\
&= \prod_{i=t+1}^{t+N} Q(Y_i|S^k_i,X^k_i)P_V(V_i)
\end{align}
and
\begin{align}
P&(Y_{t+1}^{t+N},V_{t+1}^{t+N}|y^t,v^t,s_1,W\neq k) \nonumber \\
&= \prod_{i=t+1}^{t+N} P(Y_i|Y^{i-1}_{t+1},V^{i}_{t+1},y^t,v^t,s_1,W \neq k) \nonumber \\
& \qquad P(V_i|Y^{i-1}_{t+1},V^{i-1}_{t+1},y^t,v^t,s_1,W \neq k) \nonumber \\
&=  \prod_{i=t+1}^{t+N} P_V(V_i) \sum_{x,s} Q(Y_{i}|x,s) P(X_i=x|S_i=s,Y^{i-1}_{t+1},V^{i}_{t+1},y^t,v^t,s_1,W \neq k) \nonumber \\
&\qquad P(S_i=s|Y^{i-1}_{t+1},V^{i-1}_{t+1},y^t,v^t,s_1,W \neq k) \nonumber \\
&= \prod_{i=t+1}^{t+N} P_V(V_i) \sum_{x,s} Q(Y_{i}|x,s) X^{\overline{k}}_i(x|s) B^{\overline{k}}_{i-1}(s),
\end{align}
where $X^{\overline{k}}_i(x|s)$ and $B^{\overline{k}}_{i-1}(s)$ are given by
\begin{align}
X^{\overline{k}}_i(x|s) &\eqdef  P(X_i=x|S_i=s,Y^{i-1}_{t+1},V^{i}_{t+1},y^t,v^t,s_1,W \neq k) \\
B^{\overline{k}}_{i-1}(s) &\eqdef  P(S_i=s|Y^{i-1}_{t+1},V^{i-1}_{t+1},y^t,v^t,s_1,W \neq k).
\end{align}
Moreover, $B^{\overline{k}}_i$ can be updated by
\begin{align}
B^{\overline{k}}_i(s) &= \frac{\sum_{\tilde{x},\tilde{s}} \delta_{g(\tilde{s},\tilde{x},Y_t)}(s)Q(Y_t|\tilde{x},\tilde{s})X^{\overline{k}}_i(\tilde{x}|\tilde{s})B^{\overline{k}}_{i-1}(\tilde{s})}{\sum_{\tilde{x},\tilde{s}} Q(Y_t|\tilde{x},\tilde{s})X^{\overline{k}}_i(\tilde{x}|\tilde{s})B^{\overline{k}}_{i-1}(\tilde{s})},
\end{align}
which we can concisely express as $B^{\overline{k}}_i = \phi(B^{\overline{k}}_{i-1},X^{\overline{k}}_i,Y_i)$.
With the above derivation, the divergence in~\eqref{eq:divergence} can be expressed as
\begin{align} \label{eq:divergence1}
D&(P(Y^{t+N}_{t+1},V^{t+N}_{t+1}|W=k,y^t,v^t,s_1)||P(Y^{t+N}_{t+1},V^{t+N}_{t+1}|W\neq k,y^t,v^t,s_1)) \nonumber \\
&= \sum_{i=t+1}^{t+N}\EE[ \log \frac{Q(Y_i|S^k_i,X^k_i) }{\sum_{x,s} Q(Y_{i}|x,s) X^{\overline{k}}_i(x|s) B^{\overline{k}}_{i-1}(s)} \nonumber \\
& \hspace*{6cm} |y^t,v^t,s_1,W=k] \nonumber \\
&= \sum_{i=t+1}^{t+N}\EE[ \EE[\log \frac{Q(Y_i|S^k_i,X^k_i) }{\sum_{x,s} Q(Y_{i}|x,s) X^{\overline{k}}_i(x|s) B^{\overline{k}}_{i-1}(s)} \nonumber \\
&\hspace*{1cm} |S^k_i,B^{\overline{k}}_{i-1},X^k_i,X^{\overline{k}}_i,y^t,v^t,s_1,W=k] |y^t,v^t,s_1,W=k]\nonumber \\
&= \sum_{i=t+1}^{t+N}\EE[ R(S^k_i,B^{\overline{k}}_{i-1},X^k_i,X^{\overline{k}}_i)|y^t,v^t,s_1,W=k],
\end{align}
where the function $R(s^0,b,x^0,x^1)$ is given by
\begin{align}
 R(s^0,b^1,x^0,x^1) &= \sum_y Q(y|s^0,x^0)  \nonumber \\
&\qquad \log \frac{Q(y|s^0,x^0) }{\sum_{\tilde{x},\tilde{s}} Q(y|\tilde{x},\tilde{s}) x^{1}(\tilde{x}|\tilde{s}) b^1(\tilde{s})} .
\end{align}
This inspires us to define a controlled Markov process with state $(S^0_t,B^1_{t-1})\in \cS \times \cP(\cS)$, action $(X^0_t,X^1_t) \in \cX \times (\cS \rightarrow \cP(\cX))$, instantaneous reward $R(S^0_t,B^1_{t-1},X^0_t,X^1_t)$ at time $t$ and transition kernel
\begin{align}
Q'&(S^0_{t+1},B^1_{t}|S^0_{t},B^1_{t-1},X^0_{t},X^1_{t})  \nonumber \\
 &= \sum_y \delta_{g(S^0_{t},X^0_t,y)}(S^0_{t+1})\delta_{\phi(B^1_{t-1},X^1_t,y)}(B^1_{t}) Q(y|X^0_t,S^0_{t}).
\end{align}
That this is indeed a controlled Markov process can be readily established.
Note that at time $t=0$ the process starts with initial state $(S^0_0,B^1_{-1})$.
Let $V^N(s^0,b^1)$ be the (average) reward in $N$ steps of this process
\begin{align}
V^N&(s^0,b^1) \nonumber \\
&\eqdef   \frac{1}{N} \EE[\sum_{i=1}^{N}R(S^0_i,B^1_{i-1},X^0_i,X^1_i)|S^0_0=s^0,B^1_{-1}=b^1],
\end{align}
and denote by $V^{\infty}(s^0,b^1)$ the corresponding $\limsup$, i.e.,
$V^{\infty}(s^0,b^1) =\limsup_{N\rightarrow\infty}V^N(s^0,b^1)$.
Then, the constant $C_1$ is given by
\begin{align}
C_1 &=  \sup_{s^0,b^1} V^{\infty}(s^0,b^1).
\end{align}

\subsection{A computational efficient upper bound on $C_1$ }
\label{sec:simpleC_1}
The MDP defined above has uncountably infinite state and action spaces. In this section, we propose an alternative upper bound on $C_1$ and formulate an MDP with finite state and action spaces to evaluate it. This provides a looser but more computational efficient upper bound. As it turns out, there are several instances of interest where this upper bound can be achieved~\cite{AnWu17}.
Consider again the divergence term
\begin{align}
D&( P(Y^{t+N}_{t+1},V^{t+N}_{t+1}|W=k,y^t,v^t,s_1) ||  P(Y^{t+N}_{t+1},V^{t+N}_{t+1}|W\neq k,y^t,v^t,s_1) ) \\
&=D( P(Y^{t+N}_{t+1},V^{t+N}_{t+1}|W=k,y^t,v^t,s_1) || \nonumber \\
&\qquad  \sum_{j\neq k}\frac{P(W=j|y^t,v^t,s_1)}{1-P(W=k|y^t,v^t,s_1)} P(Y^{t+N}_{t+1},V^{t+N}_{t+1}|W=j,y^t,v^t,s_1) ) \\
&\overset{(a)}{\leq}
\sum_{j\neq k}\frac{P(W=j|y^t,v^t,s_1)}{1-P(W=k|y^t,v^t,s_1)} \nonumber \\
&\qquad
D( P(Y^{t+N}_{t+1},V^{t+N}_{t+1}|W=k,y^t,v^t,s_1) || P(Y^{t+N}_{t+1},V^{t+N}_{t+1}|W=j,y^t,v^t,s_1) ) \\
&\leq \max_{j\neq k}
D( P(Y^{t+N}_{t+1},V^{t+N}_{t+1}|W=k,y^t,v^t,s_1) ||   P(Y^{t+N}_{t+1},V^{t+N}_{t+1}|W=j,y^t,v^t,s_1) )
\end{align}
where (a) is due to convexity.
Now look into the first distribution in the divergence,
\begin{align}
P(Y^{t+N}_{t+1},V^{t+N}_{t+1}|W=k,y^t,v^t,s_1) &= \prod_{i=1}^{N} P(Y_{t+i}|W=k,Y^{t+i-1}_{t+1},V^{t+i}_{t+1},y^t,v^t,s_1) P(V_{t+i}|W=k,Y^{t+i-1}_{t+1},V^{t+i-1}_{t+1},y^t,v^t,s_1)  \nonumber \\
&= \prod_{i=1}^{N} Q(Y_{t+i}|X^k_{t+i},S^k_{t+i})P_V(V_{t+i}).
\end{align}
Then we have
\begin{align}
D&(P(Y^{t+N}_{t+1},V^{t+N}_{t+1}|W=k,y^t,v^t,s_1)||P(Y^{t+N}_{t+1},V^{t+N}_{t+1}|W=j,y^t,v^t,s_1)) \nonumber \\
 &= \sum_{i=1}^{N} \EE[\EE[\log\frac{Q(Y_{t+i}|X^k_{t+i},S^k_{t+i})}{Q(Y_{t+i}|X^j_{t+i},S^j_{t+i})} \nonumber \\
&\qquad  \qquad |Y^{t+i-1}_{t+1},V^{t+i}_{t+1},y^t,v^t,s_1,W=k]|y^t,v^t,s_1,W=k] \nonumber \\
 &= \sum_{i=1}^{N} \EE[\tilde{R}(S^k_{t+i},S^j_{t+i},X^k_{t+i},X^j_{t+i})|y^t,v^t,s_1,W=k],
\end{align}
where $\tilde{R}$ is defined by
\begin{align}
\tilde{R}(s^0,s^1,x^0,x^1) = \sum_y Q(y|x^0,s^0)\log\frac{ Q(y|x^0,s^0)}{ Q(y|x^1,s^1)}.
\end{align}
Similar to the previous development, we define a controlled Markov chain with state $(S^0_t,S^1_t)\in \cS^2$, action $(X^0_t,X^1_t) \in \cX^2$ , instantaneous reward $\tilde{R}(S^0_t,S^1_t,X^0_t,X^1_t)$ at time $t$ and transition kernel
\begin{align}
\tilde{Q}'&(S^0_{t+1},S^1_{t+1}|S^0_{t},S^1_{t},X^0_{t},X^1_{t}) \nonumber \\
& = \sum_y \delta_{g(S^0_{t},X^0_t,y)}(S^0_{t+1})\delta_{g(S^1_{t},X^1_t,y)}(S^1_{t+1}) Q(y|X^0_t,S^0_{t}).
\end{align}
Let $\tilde{V}^N(s^0,s^1)$ denote the average $N$-stage reward for this MDP, i.e.,
\begin{align}
\tilde{V}^N(s^0,s^1) &\eqdef \frac{1}{N} \EE[\sum_{i=1}^{N}R(S^0_i,S^1_i,X^0_i,X^1_i)|S^0_0=s^0,S^1_0=s^1].
\end{align}
Combining the above with the definition of $C_1$, we have
\begin{align}
C_1 &\leq  \max_{s^0,s^1} \tilde{V}^{\infty}(s^0,s^1).
\end{align}
which gives an easier to evaluate upper bound on $C_1$.

\section{Numerical Result for unifilar channels}
\label{sec:example}
In this section, we provide numerical results for the expressions $V^{\infty}$ and $\tilde{V}^{\infty}$ for some binary input/output/state unifilar channels.
We consider the trapdoor channel (denoted as channel $A$), chemical channel (denoted as channel $B(p_0)$), symmetric unifilar channels (denoted as channel $C(p_0,q_0)$), and asymmetric unifilar channels (denoted as channel $C(p_0,q_0,p_1,q_1)$). All of these channels have $g(s,x,y) = s\oplus x \oplus y$ and kernel $Q$ characterized as shown in Table~\ref{t:Q}.

\begin{table}[h]
\centering
\caption{Kernel definition for binary unifilar channels}
\begin{tabular}{|c|c|c|c|c|}
\hline
Channel & $Q(0| 0,0)$& $Q(0|1,0)$ & $Q(0|0,1)$&  $Q(0|1,1)$\\
\hline
A & 1 & 0.5 & 0.5 & 0\\
\hline
B($p_0$) & $1$ & $p_0$ & $1-p_0$ & 0\\
\hline
C($p_0,q_0$) & $1-q_0$ & $p_0$ & $1-p_0$ & $q_0$\\
\hline
D($p_0,q_0,p_1,q_1$) & $1-q_0$ & $p_0$ & $1-p_1$ & $q_1$\\
\hline
\end{tabular}
\label{t:Q}
\end{table}

The numerical results are shown in the following table and were obtained by numerically solving the corresponding MDPs. The results for $V^{\infty}$ were obtained by quantizing the state and input spaces using uniform quantization with $n=100$ points. The results are tabulated in Table~\ref{t:R}.

\begin{table}[h]

\centering
\caption{Asymptotic reward per unit time}




\begin{tabular}{|c|c|c|c|c|c|c|c|}
\hline
Channel & $\inf_{s^0,b^1}V^{\infty}(s^0,b^1)$ & $\sup_{s^0,b^1}V^{\infty}(s^0,b^1)$ & $\min_{s^0,s^1}\widetilde{V}^{\infty}(s^0,s^1)$ & $\max_{s^0,s^1}\widetilde{V}^{\infty}(s^0,s^1)$ &
$C_1$ & $C_1^*$   \\
\hline
A & $\infty$ & $\infty$ & $\infty$ & $\infty$ & $\infty$ & $\infty$ \\
\hline
B($0.9$)  & $\infty$ & $\infty$ &  $\infty$ & $\infty$ & $\infty$ & $3.294$ \\
\hline
C($0.5,0.1$)&  1.633 & 1.637   & 1.637 & 1.637  & 1.637 & 1.533  \\
\hline
C($0.9,0.1$)&  2.459 & 2.536  & 2.533 & 2.536  & 2.536 & 2.459  \\
\hline
D($0.5,0.1,0.1,0.1$)& 2.274 &  2.303 & 2.298 &  2.298 & 2.303 & 2.247  \\
\hline
D($0.9,0.1,0.1,0.1$)& 2.459  & 2.536  & 2.533 &   2.536  & 2.536 & 2.459  \\
\hline
\end{tabular}

\label{t:R}
\end{table}

It is not surprising that the trapdoor and chemical channels have infinite upper bounds.
This is also true for the Z channel in the DMC case and it is related to the fact that the transition kernel has a zero entry. Intuitively, discrimination of the two hypotheses can be perfect by
transmitting always $X_t=1\oplus S_t$ under $H_0$ and $X_t=S_t$ under $H_1$ hypothesis: with high probability, that does not depend on the message size or the target error rate, the receiver under the $H_0$ hypothesis will receive the output $Y_t=1\oplus S_t$ which is impossible under $H_1$ hypothesis and thus will make a perfect decision.

For each MDP the rewards do not seem to depend on the initial state, within the accuracy of our calculations.
Similarly, the results comparing the first and second MDPs are within the accuracy of our calculations and so we cannot make a conclusive statement regarding the difference between the two MDP solutions. There is a strong indication, however, that they both result in the same average reward asymptotically.

Also shown in the above table is the quantity $C_1^*$ which is the average reward received in the MDP for the instantaneous reward
\begin{align}
R^*(s^0,b^1,x^0,x^1) = \sum_{\tilde{x},\tilde{s}} x^1(\tilde{x}|\tilde{s}) b^1(\tilde{s})\sum_y Q(y|\tilde{x},\tilde{s})\log \frac{Q(y|\tilde{x},\tilde{s})}{Q(y|x^0,s^0)}, \nonumber
\end{align}
which is of interest in the design of transmission schemes in~\cite{AnWu17}.

\section{Conclusions}
\label{sec:conclusions}
In this paper, we derive an upper bound on the error-exponent of unifilar channels with noiseless feedback and variable length codes. We generalize Burnashev's techniques by performing multi-step drift analysis and deriving a lower bound on the stopping time together with a proposed submartingale.
The constant $C_1$ which is the zero rate exponent is evaluated through an MDP and furher upper bounded through a more computationally tractable MDP. Numerical results show that for all unifilar channels tested, the two MDPs give similar results.
A future research direction is the analytical solution of these MDPs.
In addition, the presented analysis can be easily generalized to channels with finite state and inter-symbol interference (ISI) with the state known only to the receiver.

\appendices

\section{Proof of Lemma~\ref{lemma:genFano}}\label{app:genFano}
We will first establish that under the condition $P(T<\infty)=1$ the limit
$\lim_{n\rightarrow \infty} \EE[H_{T\wedge n}]$ exists, where $T\wedge n=\min\{T,n\}$.
We have
\begin{equation}
\EE[H_{T\wedge n}] = \sum_{t=1}^n \EE[H_t|T=t]P(T=t) + \EE[H_n|T>n]P(T>n).
\end{equation}
Take $m<n$ and using the fact that $0 \leq H_T \leq \log M$ a.s., we have
\begin{subequations}
\begin{align}
|\EE[H_{T\wedge n}]-\EE[H_{T\wedge m}]|
 &= \EE[H_n|T>n]P(T>n) + \EE[H_m|T>m]P(T>m) + \sum_{t=m+1}^n \EE[H_t|T=t]P(T=t) \\
 & \leq (P(T>n)+P(T>m)+\sum_{t=m+1}^n P(T=t))\log M \\
 & = 2 P(T>m) \log M \\
 & \stackrel{m\rightarrow \infty}{\longrightarrow} 0.
\end{align}
\end{subequations}
Defining the event $\cE = \{W\neq \hat{W}\}$ we have from Fano's inequality
\begin{subequations}
\begin{align}
H(W|\hat{W},T=n) &\leq h(P(\cE|T=n)) + P(\cE|T=n) \log(M-1) \Leftrightarrow  \\
\sum_{j=1}^M H(W|\hat{W}=j,T=n) P(\hat{W}=j|T=n) &\leq h(P(\cE|T=n)) + P(\cE|T=n) \log(M-1)
\label{eq:h_ineq2}
\end{align}
\end{subequations}
Now consider the probability $P(W=i|\hat{W}=j,T=n)$
\begin{subequations}
\begin{align}
P(W=i|\hat{W}=j,T=n) & = \sum_{y^n,v^n} P(W=i|\hat{W}=j,T=n,Y^n=y^n,V^n=v^n,S_1=s_1)P(Y^n=y^n,V^n=v^n,S_1=s_1|\hat{W}=j,T=n) \\
 &\stackrel{(a)}{=} \sum_{y^n,v^n} P(W=i|Y^n=y^n,V^n=v^n,S_1=s_1) P(Y^n=y^n,V^n=v^n,S_1=s_1|\hat{W}=j,T=n) \\
 &\stackrel{(b)}{=} \EE[ \Pi_n(i) | \hat{W}=j,T=n],
\end{align}
\end{subequations}
where (a) is due to $1_{T=n\text{ and } \hat{W}=j}=1_{T=n\text{ and } \hat{W}_n=j}$ being measurable wrt $\cF_n$,
and (b) is due to the definition of the rv $\Pi_n$.
Using concavity of entropy and the definition of the rv $H_n$ we now have
\begin{align}\label{eq:h_ineq1}
 \EE[ H_n | \hat{W}=j,T=n] \leq H(W|\hat{W}=j,T=n).
\end{align}
We can now write
\begin{subequations}
\begin{align}
\EE[ H_n | T=n]
 &= \sum_{j=1}^M \EE[ H_n | \hat{W}=j,T=n] P(\hat{W}=j|T=n) \\
 &\stackrel{(a)}{\leq}  \sum_{j=1}^M H(W|\hat{W}=j,T=n) P(\hat{W}=j|T=n) \\
 &\stackrel{(b)}{\leq} h(P(\cE|T=n)) + P(\cE|T=n) \log(M-1),
\end{align}
\label{eq:h_ineq3}
\end{subequations}
where (a) is due to~\eqref{eq:h_ineq1} and (b) is due to~\eqref{eq:h_ineq2}.
Averaging out wrt $T$ and using the fact that the limit $\lim_{n\rightarrow \infty} \EE[H_{T\wedge n}]$ exists,  results in
\begin{subequations}
\begin{align}
\EE[ H_T ]
 &= \sum_{n=1}^\infty \EE[H_n|T=n] P(T=n) \\
 &\stackrel{(a)}{\leq}  \sum_{n=1}^\infty h(P(\cE|T=n))P(T=n) + P(\cE|T=n)P(T=n) \log(M-1) \\
 &\stackrel{(b)}{\leq}  h(\sum_{n=1}^\infty  P(\cE|T=n)P(T=n)) + P(\cE) \log(M-1) \\
 &=  h(P_e) + P_e \log(M-1),
\end{align}
\end{subequations}
where (a) is due to~\eqref{eq:h_ineq3} and (b) is due to the concavity of the binary entropy function $h(\cdot)$.

\section{Proof of Lemma~\ref{lemma:driftentropy}}\label{app:lemma1}

Given any $y^{t}\in \cY^{t}$, $v^{t}\in \cV^{t}$ and $s_1\in \cS$,
\begin{align} \label{eq:onedrift}
\EE&[H_{t+1} - H_{t}|Y^{t}=y^{t},V^t=v^t,S_1=s_1] \nonumber\\
 &= -I(W;Y_{t+1},V_{t+1}|Y^{t}=y^{t},V^t=v^t,S_1=s_1)\nonumber \\
 &= -H(Y_{t+1}|V_{t+1},Y^{t}=y^{t},V^t=v^t,S_1=s_1)-H(V_{t+1}|Y^{t}=y^{t},V^t=v^t,S_1=s_1) \nonumber \\
 & \qquad +  H(Y_{t+1}|V_{t+1},Y^{t}=y^{t},V^t=v^t,S_1=s_1,W) + H(V_{t+1}|Y^{t}=y^{t},V^t=v^t,S_1=s_1,W) \nonumber \\
 &\overset{(a)}{=} -H(Y_{t+1}|V_{t+1},Y^{t}=y^{t},V^t=v^t,S_1=s_1)-H(V_{t+1}) \nonumber \\
 & \qquad +  H(Y_{t+1}|V_{t+1},Y^{t}=y^{t},V^t=v^t,S_1=s_1,W) + H(V_{t+1}) \nonumber \\
 &\overset{(b)}{\geq} -H(Y_{t+1}|Y^{t}=y^{t},V^t=v^t,S_1=s_1) +  H(Y_{t+1}|V_{t+1},Y^{t}=y^{t},V^t=v^t,S_1=s_1,W) \nonumber \\
 &\overset{(c)}{=} -H(Y_{t+1}|Y^{t}=y^{t},V^t=v^t,S_1=s_1) +  H(Y_{t+1}|V_{t+1},Y^{t}=y^{t},V^t=v^t,S_1=s_1,W,S_{t+1},X_{t+1}) \nonumber \\
 &\overset{(d)}{=} -H(Y_{t+1}|Y^{t}=y^{t},V^t=v^t,S_1=s_1)
  +  H(Y_{t+1}|S_{t+1},X_{t+1}, Y^{t}=y^{t},V^t=v^t,S_1=s_1) \nonumber \\
 &= -I(X_{t+1},S_{t+1};Y_{t+1}|Y^{t}=y^{t},V^t=v^t,S_1=s_1),
\end{align}
where (a) is due to the way the common random variables are selected, (b) is due to conditioning reduces entropy, and (c) is due to encoding and the deterministic channel state update (d) due to the channel properties.
Note that the last term is the mutual information between $X_{t+1},S_{t+1}$ and $Y_{t+1}$ conditioning on $Y^{t}=y^{t},V^t=v^t,S_1=s_1$, which is different from conditional mutual information $I(X_{t+1},S_{t+1};Y_{t+1}|Y^{t},V^t,S_1)$.
Now the $N$-step drift becomes
\begin{align}
\EE&[H_{t+N} - H_{t}|Y^t=y^t,V^t=v^t,S_1=s_1] \nonumber \\
 &= \sum_{k=t}^{t+N-1}  \EE[ \EE[H_{k+1} - H_{k}|Y^t=y^t,V^t=v^t, Y_{t+1}^k,V_{t+1}^k,S_1=s_1]|Y^t=y^t,V^t=v^t,S_1=s_1] \nonumber \\
 &= \sum_{k=t}^{t+N-1}  \sum_{y^k_{t+1},v^k_{t+1}} P(Y^k_{t+1}=y^k_{t+1},V^k_{t+1}=v^k_{t+1}|Y^t=y^t,V^t=v^t,S_1=s_1) \EE[H_{k+1} - H_{k}|Y^k=y^k,V^k=v^k,S_1=s_1] \nonumber \\
 &\overset{(a)}{\geq} -\sum_{k=t}^{t+N-1} \sum_{y^k_{t+1},v^k_{t+1}} P(Y^k_{t+1}=y^k_{t+1},V^k_{t+1}=v^k_{t+1}|Y^t=y^t,V^t=v^t,S_1=s_1) I(X_{k+1},S_{k+1} ;Y_{k+1}|Y^{k}=y^k,V^{k}=v^k,S_1=s_1)\nonumber \\
 &= -\sum_{k=t}^{t+N-1} I(X_{k+1},S_{k+1} ;Y_{k+1}|Y^{k}_{t+1},V^{k}_{t+1},Y^{t}=y^t,V^{t}=v^t,S_1=s_1)\nonumber \\
 &\overset{(b)}{\geq} -N(C+\epsilon),
\end{align}
where (a) is due to~\eqref{eq:onedrift} and (b) is due to~\eqref{eq:capacity}.

\section{Proof of Lemma~\ref{lemma:driftlogentropy}}\label{app:lemma2}

Given any $y^t \in \cY ^t$, $v^t \in \cV ^t$ and $s_1\in \cS $,
\begin{align}
E&[\log(H_{t+N})-\log(H_t)|Y^t=y^t,V^t=v^t,S_1=s_1] =\nonumber \\
& \EE[\log\frac{-\sum_{i}P(W=i|Y^{t+N}_{t+1},V^{t+N}_{t+1},y^t,v^t,s_1)
 \log P(W=i|Y^{t+N}_{t+1},V^{t+N}_{t+1},y^t,v^t,s_1)}{-\sum_{i}P(W=i|y^{t},v^t,s_1)
 \log P(W=i|y^t,v^t,s_1)}|Y^t=y^t,V^t=v^t,S_1=s_1].
\end{align}
For convenience, we define the following quantities
\begin{subequations}
\begin{align}
f_i &= P(W=i|y^t,v^t,s_1)  \\
f_i(Y^{t+N}_{t+1},V^{t+N}_{t+1}) &= P(W=i|Y^{t+N}_{t+1},V^{t+N}_{t+1},y^t,v^t,s_1) \\
\hat{Q}(Y^{t+N}_{t+1},V^{t+N}_{t+1}|i) &= P(Y^{t+N}_{t+1},V^{t+N}_{t+1}|W=i,y^t,v^t,s_1).
\end{align}
\end{subequations}
Since  $H_t < \epsilon$, there exits a $k$ such that $f_k>1-\epsilon/2$ while $f_j<\epsilon/2$ for $j\neq k$. We further define $\hat{f}_j \triangleq f_j/(1-f_k) $ for  $j\neq k$. The following approximations are valid for $f_k$ close to 1.
\begin{subequations}
\begin{align}
 f_k(Y^{t+N}_{t+1},V^{t+N}_{t+1}) \log f_k(Y^{t+N}_{t+1},V^{t+N}_{t+1}) &= -(1-f_k) \frac{\sum_{j\neq k} \hat{f}_j \hat{Q}(Y^{t+N}_{t+1},V^{t+N}_{t+1}|j)}{\hat{Q}(Y^{t+N}_{t+1},V^{t+N}_{t+1}|k)}    + o(1-f_k) \\
 f_j(Y^{t+N}_{t+1},V^{t+N}_{t+1}) \log f_j(Y^{t+N}_{t+1},V^{t+N}_{t+1}) &= (1-f_k) (\log(1-f_k) + o(\log(1-f_k)))\frac{\hat{f}_j\hat{Q}(Y^{t+N}_{t+1},V^{t+N}_{t+1}|j)}{\hat{Q}(Y^{t+N}_{t+1},V^{t+N}_{t+1}|k)} \\
 P(Y^{t+N}_{t+1},V^{t+N}_{t+1}|y^t,S_1) &= \hat{Q}(Y^{t+N}_{t+1},V^{t+N}_{t+1}|k)
 +o(1).
\end{align}
\end{subequations}
Substituting these approximate expressions back to the drift expression we have
\begin{align} \label{ineq:logHtlowerboundprimitive}
E&[\log (H_{t+N})-\log (H_t)|Y^t=y^t,V^t=v^t,S_1=s_1] \nonumber \\
 &= \sum_{Y^{t+N}_{t+1},V^{t+N}_{t+1}} P(Y^{t+N}_{t+1},V^{t+N}_{t+1}|y^t,v^t,s_1) \log \frac{\sum_i f_i(Y^{t+N}_{t+1},V^{t+N}_{t+1})\log f_i(Y^{t+N}_{t+1},V^{t+N}_{t+1})}{\sum_i f_i\log f_i}  \nonumber \\
 &=\sum_{Y^{t+N}_{t+1},V^{t+N}_{t+1}} \hat{Q}(Y^{t+N}_{t+1},V^{t+N}_{t+1}|k) \log \frac{(1-f_k) (\log(1-f_k) + o(\log(1-f_k)))\sum_{j\neq k}\frac{\hat{f}_j\hat{Q}(Y^{t+N}_{t+1},V^{t+N}_{t+1}|j)}{\hat{Q}(Y^{t+N}_{t+1},V^{t+N}_{t+1}|k)}}{(1-f_k)(\log (1-f_k) + o(\log(1-f_k))}   \nonumber \\
 &= -\sum_{Y^{t+N}_{t+1},V^{t+N}_{t+1}}\hat{Q}(Y^{t+N}_{t+1},V^{t+N}_{t+1}|k) \log \frac{\hat{Q}(Y^{t+N}_{t+1},V^{t+N}_{t+1}|k)}{\sum_{j\neq k}\hat{f}_j\hat{Q}(Y^{t+N}_{t+1},V^{t+N}_{t+1}|j)} +o(1)\nonumber \\
 &\geq -N(C_1 + \epsilon),
\end{align}
where the last inequality is due to the definition of $C_1$.


\section{Proof of Lemma~\ref{lemma:newsubmartingale}}\label{app:lemma3}

We can always choose a sufficiently small positive $\lambda$ such that
\begin{subequations}
\begin{align}
\frac{H^*}{K_1}(e^y-1) &< \frac{y}{K_2} + f(y) &\qquad -K_3 < y < 0 \label{ineq:subnegy} \\
\frac{H^*}{K_1}(e^y-1) &> \frac{y}{K_2} + f(y) &\qquad 0 < y < K_3  \label{ineq:subposy} \\
\frac{1}{K_2} + f'(y) &> 0  &\qquad -K_3 < y < 0. \label{ineq:subdiff} \\
1-\frac{\lambda e^{\lambda K_3}}{2K_2}K_3^2    &> 0 & \label{ineq:subsub}
\end{align}
\end{subequations}

We first consider the case $H_t > H^*$.
\begin{align}
Z_{t+1} &= (\frac{H_{t+1}-H^*}{K_1}+t+1)1_{\{H_{t+1}>H^*\}}+(\frac{\log \frac{H_{t+1}}{H^*}}{K_2}+t+1+f(\log \frac{H_t}{H^*}))1_{\{H_{t+1}\leq H^*\}}   \nonumber\\
&\overset{(a)}{\geq} (\frac{H_{t+1}-H^*}{K_1}+t+1)1_{\{H_{t+1}>H^*\}}+(\frac{H_{t+1}-H^*}{K_1}+t+1)1_{\{H_{t+1}\leq H^*\}}   \nonumber\\
& = \frac{H_{t+1}-H^*}{K_1}+t+1,
\end{align}
where (a) is due to \eqref{ineq:subnegy}. Therefore we have
\begin{align}
\EE[Z_{t+1}-Z_t|\mathcal{F}_t] &= \EE[Z_{t+1}1_{\{H_{t} > H^*\}}-Z_t1_{\{H_{t} > H^*\}}|\mathcal{F}_t] \nonumber\\
&\geq \EE[(\frac{H_{t+1}-H^*}{K_1}+t+1)1_{\{H_{t} > H^*\}}-(\frac{H_{t}-H^*}{K_1}+t)1_{\{H_{t} > H^*\}}|\mathcal{F}_t] \nonumber \\
&\geq 0, \label{ineq:subHtbig}
\end{align}
where the last equation is due to \eqref{ineq:entropy}.
Similarly, for the case $H_t \leq H^*$, from \eqref{ineq:subposy} we have
\begin{equation}
Z_{t+1} \geq (\frac{\log \frac{H_{t+1}}{H^*}}{K_2}+t+1+f(\log \frac{H_{t+1}}{H^*})),
\end{equation}
and therefore
\begin{align}
E&[Z_{t+1}-Z_t|\mathcal{F}_t] \nonumber \\
 &\geq \EE[\frac{\log \frac{H_{t+1}}{H^*}}{K_2}+t+1+f(\log \frac{H_{t+1}}{H^*})-\frac{\log \frac{H_{t}}{H^*}}{K_2}-t-f(\log \frac{H_{t}}{H^*})|\mathcal{F}_t]\nonumber \\
&\overset{(a)}{=} \EE[(\frac{1}{K_2}+f'(\log\frac{H_{t+1}}{H^*}))(\log\frac{H_{t+1}}{H^*}-\log\frac{H_{t}}{H^*})+1+\frac{f''(Z(H_{t+1},H_t))}{2}(\log\frac{H_{t+1}}{H^*}-\log\frac{H_{t}}{H^*})^2|\mathcal{F}_t]\nonumber \\
&\overset{(b)}{\geq} \EE[-K_2f'(\log\frac{H_{t+1}}{H^*})+\frac{f''(Z(H_{t+1},H_t))}{2}(\log\frac{H_{t+1}}{H^*}-\log\frac{H_{t}}{H^*})^2|\mathcal{F}_t]  \nonumber \\
&= \EE[e^{\lambda\log\frac{H_{t+1}}{H^*}}+\frac{-\lambda e^{\lambda (Z(H_{t+1},H_t)-\log\frac{H_{t+1}}{H^*}+\log\frac{H_{t+1}}{H^*})}}{2K_2}(\log\frac{H_{t+1}}{H^*}-\log\frac{H_{t}}{H^*})^2|\mathcal{F}_t] \nonumber \\
&\overset{(c)}{\geq}  \EE[e^{\lambda\log\frac{H_{t+1}}{H^*}}-\frac{\lambda e^{\lambda (K_3+\log\frac{H_{t+1}}{H^*})}}{2K_2}(\log\frac{H_{t+1}}{H^*}-\log\frac{H_{t}}{H^*})^2|\mathcal{F}_t] \nonumber \\
& \overset{(d)}{\geq} (1-\frac{\lambda e^{\lambda K_3}}{2K_2}K_3^2) \EE[e^{\lambda\log\frac{H_{t+1}}{H^*}}|\mathcal{F}_t] \nonumber \\
& \overset{(e)}{\geq}\ 0,\label{ineq:subHtsmall}
\end{align}
where (a) is from the second-order Taylor's expansion of $f$ at $\log\frac{H_t}{H^*}$, (b) is due to  \eqref{ineq:logentropysmall} and \eqref{ineq:subdiff}, (c) is due to that $Z(H_{t+1},H_t)$ is between $\log\frac{H_t}{H^*}$ and $\log\frac{H_{t+1}}{H^*}$, (d) is due to \eqref{ineq:logentropyall}, and (e) is due to \eqref{ineq:subsub}.
From \eqref{ineq:subHtbig} and \eqref{ineq:subHtsmall}, we have $\EE[Z_{t+1}-Z_t|Y^t]\geq 0$ and thus $(Z_t)_{t\geq 0}$ is a submartingale.

\section{Proof of Theorem~\ref{th:main}}\label{app:proposition}

The proof essentially applies Lemma~\ref{lemma:newsubmartingale} to the ``block'' submartingale.
Given any $\epsilon>0$, there exists an $N=N(\epsilon)$ such that by Lemma~\ref{lemma:driftentropy} and Lemma~\ref{lemma:driftlogentropy},
\begin{subequations}
\begin{align}
\EE[H_{N(t+1)}-H_{Nt}|\cF_{Nt}] &\geq  -N(C+\epsilon) \\
\EE[\log H_{N(t+1)}-\log H_{Nt}|\cF_{Nt}] &\geq -N (C_1 + \epsilon) \qquad \text{if } H_{Nt}<\epsilon \\
|\log H_{N(t+1)}-\log H_{Nt}| &<  NC_2 \qquad\qquad \text{if } H_{Nt}<\epsilon.
\end{align}
\end{subequations}
Define $M_{t'}=Z_{N t'}$, where $Z_t$ is defined in~\eqref{def:newsubmartingale}, and filtration $\cF'_{t'} = \sigma(Y^{Nt'},V^{Nt'},S_1)$. Then $(M_t')_{t'\geq 0}$ is a submartingale w.r.t. $(\cF'_{t'})_{t'\geq 0}$ by Lemma~\ref{lemma:newsubmartingale}.
Notice that the quantity $t'$ here indicates the order of the block of $N$ consecutive transmissions.
Furthermore, define the stopping time $\hat{T}$ w.r.t. $(\cF'_{t'})_{t'\geq 0}$ by $\hat{T} = \min\{k| T \leq Nk\}$. By definition of $\hat{T}$, we have
\begin{equation} \label{ineq:twostoppingtime}
(\hat{T}-1)N \leq T \qquad a.s.
\end{equation}
Now we essentially apply the optional sampling theorem on the submartingale  $(M_t')_{t'\geq 0}$ as follows
\begin{align}
\frac{K-\epsilon}{N(C+\epsilon)}& =M_0 \nonumber \\
 &\leq \EE[M_{\hat{T}}]\nonumber \\
 &= \EE[(\frac{\log H_{N\hat{T}}-\log \epsilon}{N (C_1 + \epsilon) } +f(\frac{\log H_{N\hat{T}}}{\log \epsilon}))1_{H_{N\hat{T}}\leq\epsilon}] + \EE[( \frac{H_{N\hat{T}}-\epsilon}{N(C+\epsilon)} )1_{H_{N\hat{T}}>\epsilon}] + \EE[\hat{T}] \nonumber \\
 &\leq \EE[\frac{\log H_{N\hat{T}}+|\log \epsilon|}{N (C_1 + \epsilon) } +f(\frac{\log H_{N\hat{T}}}{\log \epsilon})] + \EE[ \frac{H_{N\hat{T}}+\epsilon}{N(C+\epsilon)} ] + \EE[\hat{T}] \nonumber \\
 &\overset{(a)}{=} \EE[\frac{\log H_{T}+|\log \epsilon|}{N (C_1 + \epsilon) } +f(\frac{\log H_{T}}{\log \epsilon})] + \EE[ \frac{H_{T}+\epsilon}{N(C+\epsilon)} ] + \EE[\hat{T}]  \nonumber \\
 &\overset{(b)}{\leq} \frac{\log \EE[H_{T}]+|\log \epsilon|}{N (C_1 + \epsilon) } +\frac{1}{\lambda NC_1} +  \frac{\EE[H_{T}]+\epsilon}{N(C+\epsilon)}  + \EE[\hat{T}] \nonumber \\
 &\overset{(c)}{\leq} \frac{\log \EE[H_{T}]+|\log \epsilon|}{N (C_1 + \epsilon) } +\frac{1}{\lambda NC_1} +  \frac{\EE[H_{T}]+\epsilon}{N(C+\epsilon)}  + \frac{\EE[T]}{N}+1 \nonumber \\
 &\overset{(d)}{\leq}  \frac{\log(P_e(K-\log P_e)-\log(1-P_e) )+|\log \epsilon|}{N (C_1 + \epsilon) }  + \frac{1}{\lambda NC_1} + \frac{P_e(K-\log P_e)-\log(1-P_e) +\epsilon}{N(C+\epsilon)}+\frac{\EE[T]}{N}+1  \nonumber \\
 &\leq  \frac{\log P_e + \log(K-\log P_e)+\Delta+|\log \epsilon | }{N (C_1 + \epsilon) } +  \frac{1}{\lambda NC_1}+ \frac{P_e(K-\log P_e)-\log(1-P_e) +\epsilon}{N(C+\epsilon)} +\frac{\EE[T]}{N}+1,
\end{align}
where (a) is due to that the receiver no longer performs actions after time $T$, (b) is due the the concavity of $\log(\cdot)$ and that $f$ is upper-bounded by $\frac{1}{\lambda NC_1}$, (c) is due to \eqref{ineq:twostoppingtime}, (d) is due to the Fano's lemma and $\Delta=\log(1-\frac{\log(1-P_e)}{P_e(K-\log P_e)})$. Multiplying $N$ on the both sides the above inequality and rearranging terms, we get
\begin{align}
-\frac{\log P_e}{\EE[T]} & \leq C_1(1-\frac{\overline{R}}{C})
+ \frac{ \log(K-\log P_e)+ \Delta + |\log \epsilon|}{K/\overline{R}}  \nonumber \\
& +\frac{C_1+\epsilon}{K/\overline{R}}( \frac{1}{\lambda C_1}+\frac{-P_e \log P_e -\log(1-P_e)+2\epsilon}{C+\epsilon}+N) \nonumber \\
& + \frac{\overline{R}(C_1+\epsilon)P_e}{C+\epsilon}  + \epsilon(1-\frac{1+C_1/C}{C+\epsilon}\overline{R}). \label{ineq:result}
\end{align}
Taking the limit $K\rightarrow \infty$ on the error term on the RHS results in
$\frac{\overline{R}(C_1+\epsilon)P_e}{C+\epsilon}  + \epsilon(1-\frac{1+C_1/C}{C+\epsilon}\overline{R})$
and after taking the limit $P_e \rightarrow 0$ we get  $\epsilon(1-\frac{1+C_1/C}{C+\epsilon}\overline{R})=o_{\epsilon}(1)$.

\bibliographystyle{IEEEtran}

\end{document}